\documentclass[journal]{IEEEtran}

\usepackage[utf8]{inputenc}
\usepackage{color}
\usepackage{xcolor}
\usepackage{array}
\usepackage{verbatim}
\usepackage{float}
\usepackage{amsmath}
\usepackage{amsthm}
\usepackage{amssymb}
\usepackage{graphicx}
\usepackage{longtable}
\usepackage{multirow}
\usepackage{booktabs}
\usepackage[unicode=true,
bookmarks=false,
breaklinks=false,pdfborder={0 0 1},colorlinks=false]
{hyperref}
\hypersetup{
	colorlinks,bookmarksopen,bookmarksnumbered,citecolor=blue,urlcolor=blue}
\usepackage{cite}

\usepackage{lipsum}
\usepackage{mathtools}
\usepackage{cuted}

\usepackage{algorithmic}
\usepackage{longtable}

\floatstyle{ruled}
\newfloat{algorithm}{tbp}{loa}
\providecommand{\algorithmname}{Algorithm}
\floatname{algorithm}{\protect\algorithmname}

\makeatletter
\let\oldforeign@language\foreign@language
\DeclareRobustCommand{\foreign@language}[1]{%
	\lowercase{\oldforeign@language{#1}}}

\let\oldforeign@language\foreign@language
\DeclareRobustCommand{\foreign@language}[1]{%
	\lowercase{\oldforeign@language{#1}}}

\ifCLASSINFOpdf
\else
\fi

\newcommand{\MYfooter}{\smash{
		\hfil\parbox[t][\height][t]{\textwidth}{\centering
			\thepage}\hfil\hbox{}}}

\def\ps@IEEEtitlepagestyle{%
	\def\@oddhead{\parbox[t][\height][t]{\textwidth}{\centering \scriptsize
			Personal use of this material is permitted. Permission from the author(s) and/or copyright holder(s), must be obtained for all other uses. Please contact us and provide details if you believe this document breaches copyrights.\\
			\noindent\makebox[\linewidth]{}
		}\hfil\hbox{}}%
	\def\@evenhead{\scriptsize\thepage \hfil \leftmark\mbox{}}%
	\def\@oddfoot{\parbox[t][\height][l]{\textwidth}{
			\vspace{-20pt}{\rule{\textwidth}{0.4pt}}\\ \footnotesize\underline{To cite this article:}
			{\bf{\footnotesize\textcolor{red}{H. A. Hashim, A. E. E. Eltoukhy, K. G. Vamvoudakis, and M. I. Abouheaf "Nonlinear Deterministic Observer for Inertial Navigation using Ultra-wideband and IMU Sensor Fusion," 2023 IEEE/RSJ International Conference on Intelligent Robots and Systems (IROS), pp. 1-6, 2023.}}}\\
			\noindent\makebox[\linewidth]
		}\hfil\hbox{}}%
	\def\@evenfoot{\MYfooter}}

\makeatother
\newtheorem{lem}{Lemma}

\newtheorem{thm}{Theorem}

\newtheorem{assum}{Assumption}

\begin{document}
	\bstctlcite{IEEEexample:BSTcontrol}

	\title{Nonlinear Deterministic Observer for Inertial Navigation using Ultra-wideband and IMU Sensor Fusion}

\author{Hashim A. Hashim, Abdelrahman E. E. Eltoukhy, Kyriakos G. Vamvoudakis, and Mohammed I. Abouheaf
	\thanks{This work was supported in part by the National Sciences and Engineering Research Council of Canada (NSERC) under the grants RGPIN-2022-04937 and by the National Science Foundation under grant Nos. S\&AS-1849264, CPS-1851588, and CPS-2038589.}
	\thanks{H. A. Hashim is with the Department of Mechanical and Aerospace Engineering, Carleton University, Ottawa, ON, K1S 5B6, Canada (e-mail: hhashim@carleton.ca).
		A. E.E. Eltoukhy is with the Department of Industrial and Systems Engineering, The Hong Kong Polytechnic University, Hung Hum, Hong Kong (e-mail: abdelrahman.eltoukhy@polyu.edu.hk).	
		K. G. Vamvoudakis is with the Daniel Guggenheim School of Aerospace Engineering, Georgia Institute of Technology, Atlanta, GA, 30332, USA (e-mail: kyriakos@gatech.edu).
		M. I. Abouheaf is with the College of Technology, Architecture $\&$ Applied Engineering, Bowling Green State University, Bowling Green, OH, 43403, USA, (email: mabouhe@bgsu.edu)}
}



\maketitle

\begin{abstract}
Navigation in Global Positioning Systems (GPS)-denied environments
requires robust estimators reliant on fusion of inertial sensors able
to estimate rigid-body's orientation, position, and linear velocity.
Ultra-wideband (UWB) and Inertial Measurement Unit (IMU) represent
low-cost measurement technology that can be utilized for successful
Inertial Navigation. This paper presents a nonlinear deterministic
navigation observer in a continuous form that directly employs UWB
and IMU measurements. The estimator is developed on the extended Special
Euclidean Group $\mathbb{SE}_{2}\left(3\right)$ and ensures exponential
convergence of the closed loop error signals starting from almost
any initial condition. The discrete version of the proposed observer
is tested using a publicly available real-world dataset of a drone
flight.
\end{abstract}

\begin{IEEEkeywords}
Ultra-wideband, Inertial measurement unit, Sensor Fusion, Positioning
system, GPS-denied navigation.
\end{IEEEkeywords}

\IEEEpeerreviewmaketitle{}

\section{Introduction}

\IEEEPARstart{A}{ccurate} navigation in the absence of Global Positioning Systems (GPS)
signals is crucial for various robotics applications such as, autonomous
ground vehicles, unmanned aerial vehicles, and autonomous underwater
vehicles \cite{li2021openstreetmap,hashim2021_COMP_ENG_PRAC,zhai2020robust,zou2019structvio,hashim2021ACC}.
Common causes of GPS signal loss are multipath, obstructions, fading,
and denial in indoor environments which create the need for a backup
navigation solution. In the recent years, a number of GPS-denied navigation
solutions have been developed, for instance, vision-aided-based navigation
\cite{hashim2021_COMP_ENG_PRAC,zhai2020robust,hashim2021ACC,fornasier2022equivariant}
(monocular or stereo camera) and Light Detection and Ranging (LiDAR)-based
navigation or 3D laser scanners \cite{li2021openstreetmap}. However,
rapid advances in the areas of Micro-electromechanical systems (MEMS)
and communication technology motivates the development of navigation
solutions reliant on the fusion of Ultra-wideband (UWB) and Inertial
Measurement Unit (IMU) sensors due to their reduced price and weight,
and compactness in contrast with other aided navigation units. Moreover,
performance of the vision-based techniques degrades in low texture
environments, and both vision and LiDAR based systems are costly \cite{zou2019structvio}.
Therefore, UWB-IMU fusion could be an optimal fit for inertial navigation
of low-cost small-scale vehicles. While IMU enables rigid-body's orientation
estimation, UWB-IMU integration allows for rigid-body's position and
linear velocity estimation. Furthermore, UWB localization is possible
with Line-of-sight (LOS) and Non-line-of-sight (NLOS) communication.
\cite{yang2021novel,zihajehzadeh2015uwb,you2020data,wang2020multiple}.
However, the main challenge of UWB and IMU technology is high level
of measurement uncertainties.

Navigation based on UWB-IMU fusion requires the vehicle to be equipped
with UWB tag(s) and a 9-axis IMU (consisting of an accelerometer,
a gyroscope, and a magnetometer), along with accessibility of fixed
UWB anchors \cite{bottigliero2021low}. Since UWB and IMU measurements
are uncertain and exclude linear velocity (unlike GPS), a robust observer
design is key for the success of control missions. Recently, multiple
UWB-IMU-based filters belonging to the family of Kalman filters and
Particle Filters (PFs) have been proposed. For instance, a Kalman
Filter utilizing smooth set of coordinates compensated under NLOS
\cite{yang2021novel}, a Maximum Likelihood Kalman Filtering (MLKF)
\cite{wang2020multiple}, an extended Kalman filter \cite{bottigliero2021low},
and an Unscented Kalman Filter (UKF) neglecting the high-order terms
\cite{you2020data}. PFs are commonly classified as stochastic filters
\cite{tian2019resetting}, have been introduced to improve estimation
accuracy and address the consistency issue associated with Kalman-type
filters. The limitation of the above-mentioned Kalman-type filters
\cite{wang2020multiple,you2020data} is the reliance on linearization
around a nominal point ignoring high order nonlinear terms and lowering
the estimation accuracy \cite{kallianpur2013stochastic}. Moreover,
UKF utilizes a set of sigma points complicating filter design and
implementation. Meanwhile, PFs are challenged with higher computational
cost requirements and lack of a clear measure of optimal performance
\cite{kallianpur2013stochastic}. Note that state-of-the-art UWB-IMU-based
navigation filters rely on Euler angels which are subject to singularities
\cite{lefferts1982kalman} in particular for a rigid-body rotating
in three-dimensional (3D) space. Consequently, robust and accurate
navigation algorithms for GPS-denied environments remain a challenging
open problem.

\paragraph*{Contributions}This work aims to frame the navigation
kinematics on the Lie group of the extended Special Euclidean Group
$\mathbb{SE}_{2}\left(3\right)$. In this work, we consider a vehicle
equipped with a 9-axis IMU and at least one UWB tag navigating within
the range of fixed UWB anchors. A nonlinear deterministic navigation
observer on the Lie group of $\mathbb{SE}_{2}\left(3\right)$ reliant
on UWB and IMU measurements is proposed. The proposed observer successfully
addresses the unknown bias present in IMU measurements. The proposed
observer is tested using a publicly available real-world drone flight
dataset \cite{zhao2022uwbData}.

The remainder of the paper is organized as follows: Section \ref{sec:Preliminaries-and-Math}
contains preliminaries and mathematical notation. Section \ref{sec:SE3_Problem-Formulation}
formulates the problem. Section \ref{sec:UWB_Filter} introduces the
proposed nonlinear navigation observer on $\mathbb{SE}_{2}\left(3\right)$.
In Section \ref{sec:UWB_Simulations}, the proposed navigation observer
is validated using a real-world drone flight dataset. Finally, Section
\ref{sec:SE3_Conclusion} summarizes the work.

\section{Preliminaries and Math Notation \label{sec:Preliminaries-and-Math}}

In this paper, $\mathbb{R}^{a}$, $\mathbb{R}^{a\times b}$, and $\mathbb{R}_{+}$
stands for the set of $a$ dimensional Euclidean space, an $a$-by-$b$
dimensional space, and a set of nonnegative real numbers, respectively.
The Euclidean norm of $x\in\mathbb{R}^{n}$ is described by $||x||=\sqrt{x^{\top}x}$
while the Frobenius norm of $M$ is represented by $||M||_{F}=\sqrt{{\rm Tr}\{MM^{*}\}}$
with $*$ referring to a conjugate transpose. The $m$-by-$m$ identity
matrix is described by $\mathbf{I}_{n}$ and the $m$-by-$n$ zero
matrix is denoted as $0_{n\times m}$. The set of eigenvalues of $M_{r}\in\mathbb{R}^{n\times n}$
is denoted as $\lambda(M_{r})=\{\lambda_{1},\lambda_{2},\ldots,\lambda_{n}\}$.
For $M_{r}\in\mathbb{R}^{n\times n}$, $\overline{\lambda}_{M_{r}}=\overline{\lambda}(M_{r})$
and $\underline{\lambda}_{M_{r}}=\underline{\lambda}(M_{r})$ describe
the maximum and the minimum eigenvalues of $\lambda(M_{r})$, respectively.
For a vehicle navigating with six degrees of freedom (6 DoF), let
us denote $\left\{ \mathcal{I}\right\} $ as the fixed inertial-frame
and $\left\{ \mathcal{B}\right\} $ as the fixed body-frame. $\mathbb{SO}\left(3\right)$
denotes the Special Orthogonal Group $\mathbb{SO}\left(3\right)$
where \cite{hashim2018SO3Stochastic,hashim2019SO3Wiley}
\[
\mathbb{SO}(3)=\{R\in\mathbb{R}^{3\times3}|R^{\top}R=\mathbf{I}_{3}\text{, }{\rm det}(R)=+1\}
\]
with $R\in\mathbb{SO}\left(3\right)$ being rigid-body's orientation
known as attitude. $\mathfrak{so}(3)$ describes the Lie algebra of
$\mathbb{SO}(3)$ defined as
\begin{align*}
	\mathfrak{so}(3) & =\{[y]_{\times}\in\mathbb{R}^{3\times3}|y\in\mathbb{R}^{3}\}\\
	\left[y\right]_{\times} & =\left[\begin{array}{ccc}
		0 & -y_{3} & y_{2}\\
		y_{3} & 0 & -y_{1}\\
		-y_{2} & y_{1} & 0
	\end{array}\right]\in\mathfrak{so}\left(3\right),\hspace{1em}y=\left[\begin{array}{c}
		y_{1}\\
		y_{2}\\
		y_{3}
	\end{array}\right]
\end{align*}
with $\left[y\right]_{\times}^{\top}=-\left[y\right]_{\times}$ being
a skew symmetric matrix{\small{}.} The inverse mapping of $[\cdot]_{\times}$
a 3-dimensional vector ($\mathbf{vex}:\mathfrak{so}\left(3\right)\rightarrow\mathbb{R}^{3}$)
defined by
\begin{align*}
	\mathbf{vex}([y]_{\times})= & y,\forall y\in\mathbb{R}^{3}\\
	\boldsymbol{\mathcal{P}}_{a}(Y)= & \frac{1}{2}(Y-Y^{\top})\in\mathfrak{so}(3),\forall Y\in\mathbb{R}^{3\times3}
\end{align*}
and $\mathbf{vex}(\boldsymbol{\mathcal{P}}_{a}(Y))=\frac{1}{2}[Y_{3,2}-Y_{2,3},Y_{1,3}-Y_{3,1},Y_{2,1}-Y_{1,2}]^{\top}\in\mathbb{R}^{3}$.
Define the normalized Euclidean distance of $R\in\mathbb{SO}(3)$
as
\begin{equation}
	||R||_{{\rm I}}=\frac{1}{4}{\rm Tr}\{\mathbf{I}_{3}-R\}\in\left[0,1\right]\label{eq:NAV_Ecul_Dist}
\end{equation}
where $-1\leq{\rm Tr}\{R\}\leq3$. For $M\in\mathbb{R}^{3\times3}$,
$||MR||_{{\rm I}}=\frac{1}{4}{\rm Tr}\{M-MR\}$. For a rigid-body
traveling with 6 DoF, let $R\in\mathbb{SO}\left(3\right)$, $P\in\mathbb{R}^{3}$,
and $V\in\mathbb{R}^{3}$ denote the rigid-body's true orientation,
position, and velocity, respectively, where $R\in\{\mathcal{B}\}$
and $P,V\in\{\mathcal{I}\}$. Consider the extended form of the Special
Euclidean Group $\mathbb{SE}_{2}\left(3\right)=\mathbb{SO}\left(3\right)\times\mathbb{R}^{3}\times\mathbb{R}^{3}\subset\mathbb{R}^{5\times5}$
\cite{barrau2016invariant}
\begin{align}
	\mathbb{SE}_{2}(3) & =\{\left.X\in\mathbb{R}^{5\times5}\right|R\in\mathbb{SO}\left(3\right),P,V\in\mathbb{R}^{3}\}\label{eq:NAV_SE2_3}\\
	X=\Psi( & R,P,V)=\left[\begin{array}{ccc}
		R & P & V\\
		0_{1\times3} & 1 & 0\\
		0_{1\times3} & 0 & 1
	\end{array}\right]\in\mathbb{SE}_{2}\left(3\right)\label{eq:NAV_X}
\end{align}
where $X\in\mathbb{SE}_{2}\left(3\right)$ refers to the homogeneous
navigation matrix. Let us define $\Omega\in\mathbb{R}^{3}$, $V\in\mathbb{R}^{3}$,
and $a\in\mathbb{R}^{3}$ as the rigid-body's true angular velocity,
linear velocity, and acceleration, respectively, with $\Omega,a\in\{\mathcal{B}\}$.
Let us define the submanifold $\mathcal{U}_{\mathcal{M}}=\mathfrak{so}(3)\times\mathbb{R}^{3}\times\mathbb{R}^{3}\times\mathbb{R}\subset\mathbb{R}^{5\times5}$
as
\begin{align}
	\mathcal{U}_{\mathcal{M}} & =\{\left.u([\Omega\text{\ensuremath{]_{\times}}},V,a,\kappa)\right|[\Omega\text{\ensuremath{]_{\times}}}\in\mathfrak{so}(3),V,a\in\mathbb{R}^{3},\varrho\in\mathbb{R}\}\nonumber \\
	u( & [\Omega\text{\ensuremath{]_{\times}}},V,a,\kappa)=\left[\begin{array}{ccc}
		[\Omega\text{\ensuremath{]_{\times}}} & V & a\\
		0_{1\times3} & 0 & 0\\
		0_{1\times3} & \varrho & 0
	\end{array}\right]\in\mathcal{U}_{\mathcal{M}}\subset\mathbb{R}^{5\times5}\label{eq:NAV_u}
\end{align}
To know more about $\mathbb{SE}_{2}\left(3\right)$ and $\mathcal{U}_{\mathcal{M}}$
visit \cite{hashim2021_COMP_ENG_PRAC,hashim2021ACC}. 

\section{UWB, IMU, and Navigation\label{sec:SE3_Problem-Formulation}}

The UWB sensors have short wavelength which increases positioning
accuracy and their robustness against interference and fading (well-known
shortcomings of GPS communication) \cite{zihajehzadeh2015uwb,you2020data}.
UWB sensors are capable of LOS and NLOS communication and obstacle
penetration. Furthermore, UWB technology is low in power consumption,
compact, and light-weight warranting ease of implementation. Thus,
UWB sensors are fit for a positioning system as long as a robust estimation
algorithm able to reject uncertainties and produce a reasonable position
estimate is employed. UWB positioning can be achieved through various
techniques, such as Time Of Arrival (TOA), Angle of Arrival (AOA),
Time Difference Of Arrival (TDOA), and Received Signal Strength (RSS)
\cite{yang2021novel,zihajehzadeh2015uwb}. These approaches are very
close in concept. Generally, a UWB tag attached to a vehicle allows
to position it using range difference between several Base Stations
(BSs) \cite{wang2020multiple}. This work employs the more practical
and common TDOA technique. To implement TDOA, let us define $d_{j,i}\in\mathbb{R}$
as the range distance at the UWB tag, and $P=[x,y,z]^{\top}\in\mathbb{R}^{3}$
as the vehicle's position (with an attached UWB tag). The difference
in signals received from the $i$th fixed anchor $h_{i}=[x_{i},y_{i},z_{i}]^{\top}\in\mathbb{R}^{3}$
and the $j$th fixed anchor $h_{j}=[x_{j},y_{j},z_{j}]^{\top}\in\mathbb{R}^{3}$
is defined by
\begin{align}
	d_{j,i}= & ||P-h_{j}||-||P-h_{i}||\label{eq:UWB_dji}
\end{align}
The equation in \eqref{eq:UWB_dji} can be squared showing that 
\[
\frac{d_{j,i}^{2}+||h_{i}||^{2}-||h_{j}||^{2}}{2}=(h_{i}-h_{j})^{\top}P-d_{j,i}||P-h_{i}||
\]
In view of \eqref{eq:UWB_dji}, and considering $N$ TDOA measurements,
the following expression can be obtained:\begin{small}
	\begin{align}
		\frac{d_{2,1}^{2}+||h_{1}||^{2}-||h_{2}||^{2}}{2}= & (h_{1}-h_{2})^{\top}P-d_{2,1}||P-h_{1}||\nonumber \\
		\frac{d_{3,2}^{2}+||h_{2}||^{2}-||h_{3}||^{2}}{2}= & (h_{2}-h_{3})^{\top}P-d_{3,2}||P-h_{2}||\nonumber \\
		\vdots\nonumber \\
		\frac{d_{1,N}^{2}+||h_{N}||^{2}-||h_{1}||^{2}}{2}= & (h_{N}-h_{1})^{\top}P-d_{1,N}||P-h_{N}||\label{eq:UWB_djin}
	\end{align}
\end{small}Considering $||P-h_{3}||=d_{3,2}+||P-h_{2}||$, one finds
\begin{align*}
	||P-h_{3}|| & =d_{3,2}+d_{2,1}+||P-h_{1}||\\
	||P-h_{4}|| & =d_{4,3}+d_{3,2}+d_{2,1}+||P-h_{1}||
\end{align*}
As such, for $N$ TDOA measurements, one shows
\begin{align*}
	||P-h_{N}|| & =\sum_{i=2}^{N}d_{i,i-1}+||P-h_{1}||
\end{align*}
Define the following matrices
\[
A=\left[\begin{array}{cc}
	(h_{1}-h_{2})^{\top} & -d_{2,1}\\
	(h_{2}-h_{3})^{\top} & -d_{3,2}\\
	\vdots & \vdots\\
	(h_{N-1}-h_{N})^{\top} & -d_{N,N-1}\\
	(h_{N}-h_{1})^{\top} & -d_{1,N}
\end{array}\right]
\]
and
\[
B=\frac{1}{2}\left[\begin{array}{c}
	d_{2,1}^{2}+||h_{1}||^{2}-||h_{2}||^{2}\\
	d_{3,2}^{2}+||h_{2}||^{2}-||h_{3}||^{2}+2d_{3,2}\sum_{i=2}^{2}d_{i,i-1}\\
	d_{4,3}^{2}+||h_{3}||^{2}-||h_{4}||^{2}+2d_{4,3}\sum_{i=2}^{3}d_{i,i-1}\\
	\vdots\\
	d_{1,N}^{2}+||h_{N}||^{2}-||h_{1}||^{2}+2d_{1,N}\sum_{i=2}^{N}d_{i,i-1}
\end{array}\right]
\]
with $N$ being the number of fixed anchors or BSs accessed by the
tag. Hence, one obtains $A\overline{P}-B=0$ where $\overline{P}=[P^{\top},||P-h_{1}||]^{\top}\in\mathbb{R}^{4}$.
Thus, by defining $\delta=\frac{1}{2}(AP-B)^{\top}(AP-B)$ and applying
minimum mean square error, one obtains $\frac{\partial\delta}{\partial P}=A^{\top}(AP-B)=0$
such that
\begin{equation}
	\overline{P}=(A^{\top}A)^{-1}A^{\top}B\label{eq:UWB_Pbar}
\end{equation}
where $\overline{P}=[P^{\top},||P-h_{1}||]^{\top}\in\mathbb{R}^{4}$.

\begin{assum}\label{Assumption:assum_TOA} The rigid-body's position
	$P=[x,y,z]^{\top}$ can be uniquely defined in 3D space if at each
	time instant the tag is within range of at least 4 anchors. Analogously,
	3 or more are sufficient to position the rigid-body in 2D space.\end{assum}

The 9-axis IMU consists of three units: a gyroscope, an accelerometer,
and a magnetometer \cite{hashim2020SE3Stochastic,kang2019unscented,stovner2018attitude}.
The gyro supplies measurements of rigid-body's angular velocity expressed
as follows:
\begin{equation}
	\Omega_{m}=\Omega+b_{\Omega}\in\mathbb{R}^{3}\label{eq:UWB_Angular}
\end{equation}
with $\Omega$ being the true angular velocity and $b_{\Omega}$ referring
to unknown bias. The accelerometer provides acceleration measurements:
\begin{equation}
	a_{m}=R^{\top}(\dot{V}-\overrightarrow{\mathtt{g}})+b_{a}\in\mathbb{R}^{3}\label{eq:UWB_am}
\end{equation}
with $\overrightarrow{\mathtt{g}}=[0,0,g]^{\top}$, $g=-9.8\text{m}/\text{sec}^{2}$
denoting gravitational acceleration, and $\dot{V}$ referring to linear
acceleration. $b_{a}$ represents unknown bias. At low frequency,
$|\overrightarrow{\mathtt{g}}|>>|\dot{V}|$. Hence, $a_{m}$ can be
approximated by $a_{m}\approx-R^{\top}\overrightarrow{\mathtt{g}}+n_{a}$.
The magnetometer measurements are defined by
\begin{equation}
	m_{m}=R^{\top}m_{r}+b_{m}\in\mathbb{R}^{3}\label{eq:UWB_Mm}
\end{equation}
with $m_{r}=[m_{N},0,m_{D}]^{\top}$ being the earth-magnetic field
and $b_{m}$ referring to unknown bias. Three non-collinear observations
and measurements are necessary for attitude observation commonly obtained
as follows:
\begin{equation}
	\begin{cases}
		v_{1}=\frac{a_{m}}{||a_{m}||}, & r_{1}=\frac{-\overrightarrow{\mathtt{g}}}{||-\overrightarrow{\mathtt{g}}||}\\
		v_{2}=\frac{m_{m}}{||m_{m}||}, & r_{2}=\frac{m_{r}}{||m_{r}||}\\
		v_{3}=\frac{v_{1}\times v_{2}}{||v_{1}\times v_{2}||}, & r_{3}=\frac{r_{1}\times r_{2}}{||r_{1}\times r_{2}||}
	\end{cases}\label{eq:UWB_IMU_Normal}
\end{equation}
To this end, the true navigation kinematics of a rigid-body traveling
with 6 DoF are as follows \cite{hashim2021_COMP_ENG_PRAC,fornasier2022equivariant,hashim2021ACC}:
\begin{equation}
	\begin{cases}
		\dot{R} & =R\left[\Omega\right]_{\times}\\
		\dot{P} & =V\\
		\dot{V} & =Ra+\overrightarrow{\mathtt{g}}
	\end{cases},\hspace{0.5em}\underbrace{\dot{X}=XU-\mathcal{\mathcal{G}}X}_{\text{Compact form}}\label{eq:UWB_NAV_dot}
\end{equation}
where $R\in\mathbb{SO}\left(3\right)$ stands for the true orientation,
$P\in\mathbb{R}^{3}$ describe the true position, $V\in\mathbb{R}^{3}$
expresses the true linear velocity, $\Omega\in\mathbb{R}^{3}$ denotes
the true angular velocity, and $a\in\mathbb{R}^{3}$ stands for the
acceleration for all $R,\Omega,a\in\{\mathcal{B}\}$ and $P,V\in\{\mathcal{I}\}$.
The right part of \eqref{eq:UWB_NAV_dot} constitutes the compact
form of the navigation kinematics where $X=\Psi(R,P,V)\in\mathbb{SE}_{2}\left(3\right)$
(see the map in \eqref{eq:NAV_u}), $U=u([\Omega\text{\ensuremath{]_{\times}}},0_{3\times1},a,1)\in\mathcal{U}_{m}$,
and $\mathcal{\mathcal{G}}=u(0_{3\times3},0_{3\times1},-\overrightarrow{\mathtt{g}},1)\in\mathcal{U}_{m}$
(see the map in \eqref{eq:NAV_u}). For more information visit \cite{hashim2021_COMP_ENG_PRAC}.
\begin{lem}
	\label{Lemm:UWB_Lemma2}\cite{hashim2019SO3Wiley} Define $R\in\mathbb{SO}\left(3\right)$,
	$M_{r}=M_{r}^{\top}\in\mathbb{R}^{3\times3}$, and consider $\overline{M_{r}}={\rm Tr}\{M_{r}\}\mathbf{I}_{3}-M_{r}$
	where $\overline{\lambda}_{\overline{M_{r}}}$ and $\underline{\lambda}_{\overline{M_{r}}}$
	stands for the minimum and the maximum eigenvalues of $\overline{M_{r}}$,
	respectively. Define $||M_{r}R||_{{\rm I}}=\frac{1}{4}{\rm Tr}\{M_{r}(\mathbf{I}_{3}-R)\}$.
	As such, one obtains:
	\begin{align}
		||\mathbf{vex}(\boldsymbol{\mathcal{P}}_{a}(M_{r}R))||^{2} & \leq2\overline{\lambda}_{\overline{M_{r}}}||M_{r}R||_{{\rm I}}\label{eq:UWB_lemm2_1}\\
		||\mathbf{vex}(\boldsymbol{\mathcal{P}}_{a}(M_{r}R))||^{2} & \geq\frac{\underline{\lambda}_{\overline{M_{r}}}}{2}||M_{r}R||_{{\rm I}}(1+{\rm Tr}\{R\})\label{eq:UWB_lemm2_2}
	\end{align}
\end{lem}

\section{Deterministic Navigation Observer \label{sec:UWB_Filter}}

Let us define the estimates of the rigid-body's orientation, position,
and linear velocity as $\hat{R}\in\mathbb{SO}\left(3\right)$, $\hat{P}\in\mathbb{R}^{3}$,
and $\hat{V}\in\mathbb{R}^{3}$, respectively. The aim of this section
is to propose a nonlinear deterministic navigation observer reliant
on UWB-IMU fusion, that drives $\hat{R}\rightarrow R$, $\hat{P}\rightarrow P,$and
$\hat{V}\rightarrow V$. Fig. \ref{fig:TDOA1} presents a conceptual
illustration of the navigation problem, UWB-IMU fusion, and the estimation
objective.

\begin{figure}[h]
	\centering{}\centering\includegraphics[scale=0.45]{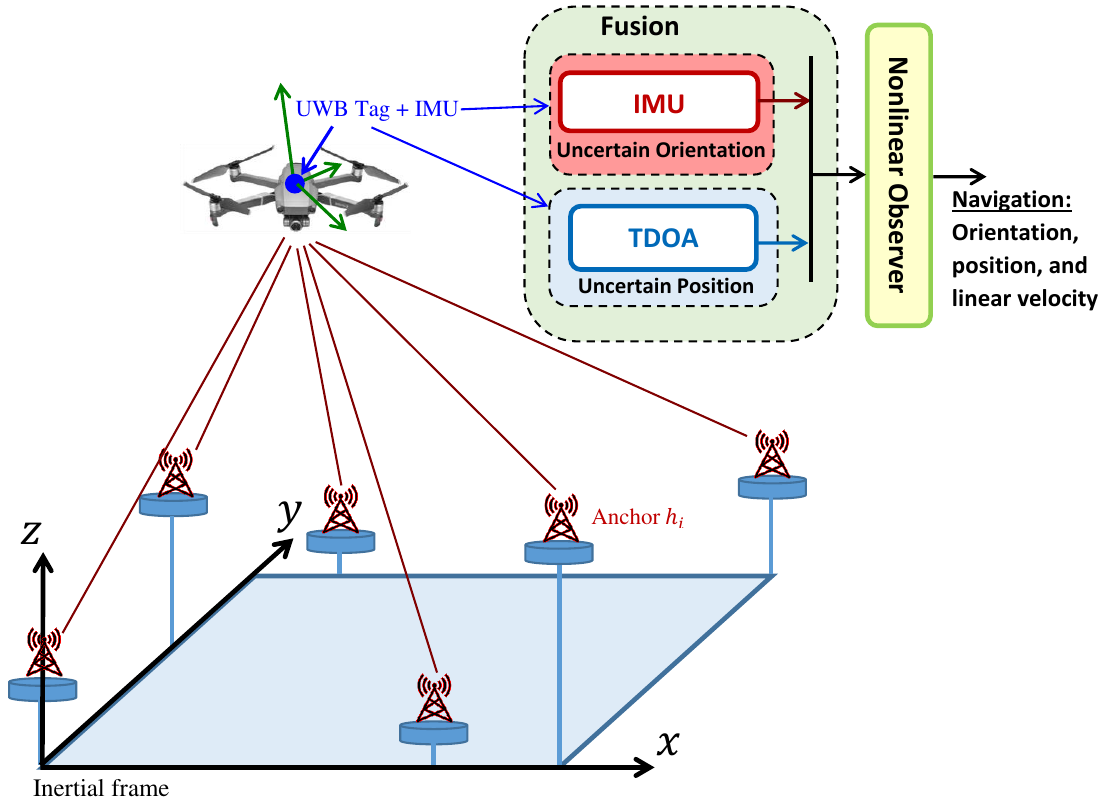}\caption{\label{fig:TDOA1} UWB-IMU fusion and navigation estimation problem.}
\end{figure}
Denote the estimation errors in attitude, position, and linear velocity
as $\tilde{R}$ , $\tilde{P}$, and $\tilde{V}$, respectively, and
express them as follows:
\begin{equation}
	\begin{cases}
		\tilde{R} & =R\hat{R}^{\top}\\
		\tilde{P} & =P-\hat{P}\\
		\tilde{V} & =V-\hat{V}
	\end{cases}\label{eq:UWB_NAV_error}
\end{equation}
Denote bias estimates of angular velocity and accelerometer as $\hat{b}_{\Omega}$
and $\hat{b}_{a}$, respectively, and express the bias estimation
error as follows:
\begin{equation}
	\begin{cases}
		\tilde{b}_{\Omega} & =b_{\Omega}-\hat{b}_{\Omega}\in\mathbb{R}^{3}\\
		\tilde{b}_{a} & =b_{a}-\hat{b}_{a}\in\mathbb{R}^{3}
	\end{cases}\label{eq:UWB_NAV_s_error}
\end{equation}
Recalling \eqref{eq:UWB_IMU_Normal} and define 
\begin{align}
	M_{r} & =\sum_{i=1}^{3}s_{i}r_{i}r_{i}^{\top},\hspace{1em}M_{B}=\sum_{i=1}^{3}s_{i}v_{i}v_{i}^{\top}\label{eq:UWB_Mr_Mv}
\end{align}
with $s_{i}$ standing for the $i$th sensor confidence level and
$\sum_{i=1}^{3}s_{i}=3$. Let us define
\begin{equation}
	\hat{v}_{i}=\hat{R}^{\top}r_{i},\hspace{1em}\forall i=1,2,3\label{eq:UWB_v_est}
\end{equation}
Hence, one obtains
\begin{align}
	\mathbf{vex}(\boldsymbol{\mathcal{P}}_{a}(M_{r}\tilde{R})) & =\frac{1}{2}\mathbf{vex}(M_{r}\tilde{R}-\tilde{R}^{\top}M_{r})\nonumber \\
	& =\frac{1}{2}\mathbf{vex}\left(\sum_{i=1}^{3}s_{i}r_{i}v_{i}^{\top}\hat{R}^{\top}-\sum_{i=1}^{3}s_{i}\hat{R}v_{i}r_{i}^{\top}\right)\nonumber \\
	& =\frac{1}{2}\sum_{i=1}^{3}\hat{R}s_{i}(v_{i}\times\hat{v}_{i})\label{eq:UWB_VEX}
\end{align}
Note that $[v_{i}\times\hat{v}_{i}]_{\times}=\hat{v}_{i}v_{i}^{\top}-v_{i}\hat{v}_{i}^{\top}$.

\paragraph*{Nonlinear observer}Let us introduce the following correction
mechanism:
\begin{equation}
	\begin{cases}
		\overline{P}_{y} & =\left[\begin{array}{c}
			P_{y}\\
			||P-h_{1}||
		\end{array}\right]=(A^{\top}A)^{-1}A^{\top}B\\
		\dot{\hat{b}}_{\Omega} & =-\frac{\gamma_{\Omega}}{2}\sum_{i=1}^{n}(v_{i}\times\hat{v}_{i})\\
		w_{\Omega} & =-\frac{k_{\Omega}}{2}\sum_{i=1}^{n}\hat{R}(v_{i}\times\hat{v}_{i})\\
		w_{V} & =-k_{v}(P_{y}-\hat{P})-[w_{\Omega}]_{\times}\hat{P}\\
		w_{a} & =-k_{a}(P_{y}-\hat{P})-[w_{\Omega}]_{\times}\hat{V}
	\end{cases}\label{eq:UWB_Filter1_Correc}
\end{equation}
where $\gamma_{\sigma}$, $k_{\Omega}$, $k_{v}$, and $k_{a}$, are
positive constants, $v_{i}$ is defined in \eqref{eq:UWB_IMU_Normal},
$v_{i}$ is described in \eqref{eq:UWB_v_est}, and $P_{y}$ is expressed
in \eqref{eq:UWB_Pbar}. Consider the following nonlinear deterministic
navigation observer design:
\begin{equation}
	\begin{cases}
		\dot{\hat{R}} & =\hat{R}\left[\Omega_{m}-\hat{b}_{\Omega}\right]_{\times}-\left[w_{\Omega}\right]_{\times}\hat{R}\\
		\dot{\hat{P}} & =\hat{V}-\left[w_{\Omega}\right]_{\times}\hat{P}-w_{V}\\
		\dot{\hat{V}} & =\hat{R}a_{m}+\overrightarrow{\mathtt{g}}-\left[w_{\Omega}\right]_{\times}\hat{V}-w_{a}
	\end{cases},\hspace{1em}\underbrace{\dot{\hat{X}}=\hat{X}U_{m}-W\hat{X}}_{\text{Compact form}}\label{eq:UWB_Filter1_Detailed}
\end{equation}
where $w_{\Omega}$, $w_{V}$, and $w_{a}$ are defined in \eqref{eq:UWB_Filter1_Correc}.
The right part of \eqref{eq:UWB_Filter1_Detailed} comprises the compact
form of the estimator kinematics where $\hat{X}=\Psi(\hat{R},\hat{P},\hat{V})\in\mathbb{SE}_{2}\left(3\right)$
describes the homogeneous navigation estimate of $X$ (see the map
in \eqref{eq:NAV_X}), $U_{m}=u([\Omega_{m}-\hat{b}_{\Omega}\text{\ensuremath{]_{\times}}},0_{3\times1},a_{m},1)\in\mathcal{U}_{\mathcal{M}}$,
and $W=u([w_{\Omega}\text{\ensuremath{]_{\times}}},w_{V},w_{a},1)\in\mathcal{U}_{\mathcal{M}}$
(see the map in \eqref{eq:NAV_u}).
\begin{thm}
	\label{thm:Theorem1} Consider the nonlinear navigation system in
	\eqref{eq:UWB_NAV_dot}. Assume availability of 3 non-collinear measurements/observations
	and fulfillment of Assumption \ref{Assumption:assum_TOA}. Let the
	nonlinear deterministic navigation observer in \eqref{eq:UWB_Filter1_Detailed}
	be coupled with the direct measurements in \eqref{eq:UWB_Pbar} and
	the correction terms in \eqref{eq:UWB_Filter1_Correc} such that $\Omega_{m}=\Omega+b_{\Omega}$.
	Hence, all the closed-loop signals are exponentially stable from almost
	any initial condition.
\end{thm}
\begin{proof}In view of \eqref{eq:UWB_NAV_dot}, \eqref{eq:UWB_NAV_error},
	and \eqref{eq:UWB_Filter1_Detailed}, one shows
	\begin{align}
		\frac{d}{dt}||M_{r}\tilde{R}||_{{\rm I}}= & \frac{d}{dt}\frac{1}{4}{\rm Tr}\{M_{r}(\mathbf{I}_{3}-\tilde{R})\}\nonumber \\
		= & \frac{1}{4}{\rm Tr}\{M_{r}\tilde{R}[\hat{R}\tilde{b}_{\Omega}-w_{\Omega}]_{\times}\}\nonumber \\
		= & -\frac{1}{2}\mathbf{vex}(\boldsymbol{\mathcal{P}}_{a}(M_{r}\tilde{R}))^{\top}(\hat{R}\tilde{b}_{\Omega}-w_{\Omega})\label{eq:UWB_Er_dot}
	\end{align}
	where $M_{r}$ is a constant matrix and 
	\begin{align*}
		{\rm Tr}\{M_{r}\tilde{R}[w_{\Omega}]_{\times}\} & ={\rm Tr}\{\boldsymbol{\mathcal{P}}_{a}(M_{r}\tilde{R})[w_{\Omega}]_{\times}\}\\
		& =-\frac{1}{2}\mathbf{vex}(\boldsymbol{\mathcal{P}}_{a}(M_{r}\tilde{R}))^{\top}w_{\Omega}
	\end{align*}
	From \eqref{eq:UWB_NAV_dot}, \eqref{eq:UWB_NAV_error}, and \eqref{eq:UWB_Filter1_Detailed},
	one finds
	\begin{equation}
		\begin{cases}
			\dot{\tilde{P}} & =\tilde{V}+[w_{\Omega}]_{\times}\hat{P}+w_{V}\\
			d\tilde{V} & =(\tilde{R}-\mathbf{I}_{3})\hat{R}a+[w_{\Omega}]_{\times}\hat{V}+w_{a}
		\end{cases}\label{eq:UWB_Filter1_Error_dot}
	\end{equation}
	Consider the Lyapunov function candidate $\mathcal{L}_{T}=\mathcal{L}_{T}(E_{r},\tilde{P},\tilde{V},\tilde{b}_{\Omega})$:
	\begin{equation}
		\mathcal{L}_{T}=\mathcal{L}_{R}+\mathcal{L}_{PV}\label{eq:UWB_LyapT}
	\end{equation}
	Define the following Lyapunov function candidate $L_{1}:\mathbb{SO}\left(3\right)\times\mathbb{R}^{3}\rightarrow\mathbb{R}_{+}$:
	\begin{equation}
		L_{1}=2||M_{r}\tilde{R}||_{{\rm I}}+\frac{1}{2\gamma_{\Omega}}||\tilde{b}_{\Omega}||^{2}\label{eq:UWB_LyapR}
	\end{equation}
	From \eqref{eq:UWB_LyapR} and \eqref{eq:UWB_Filter1_Detailed}, one
	shows
	\begin{align}
		\dot{L}_{1} & =-\mathbf{vex}(\boldsymbol{\mathcal{P}}_{a}(M_{r}\tilde{R}))^{\top}(\hat{R}\tilde{b}_{\Omega}-w_{\Omega})-\frac{1}{\gamma_{\Omega}}\tilde{b}_{\Omega}^{\top}\dot{\hat{b}}_{\Omega}\nonumber \\
		& =-k_{\Omega}||\mathbf{vex}(\boldsymbol{\mathcal{P}}_{a}(M_{r}\tilde{R}))||^{2}\label{eq:UWB_LyapR1dot}
	\end{align}
	$||\mathbf{vex}(\boldsymbol{\mathcal{P}}_{a}(M_{r}\tilde{R}))||\rightarrow0_{3\times1}$
	shows that $w_{\Omega}\rightarrow0_{3\times1}$ and $\dot{\hat{b}}_{\Omega}\rightarrow0_{3\times1}$.
	Hence, $\dot{\tilde{R}}\rightarrow0_{3\times3}$ leading to $w_{\Omega}-\hat{R}\tilde{b}_{\Omega}\rightarrow0_{3\times1}$,
	and thereby $\tilde{b}_{\Omega}\rightarrow0_{3\times1}$. Since $\boldsymbol{\Upsilon}(M_{r}\dot{R})=\frac{1}{2}({\rm Tr}\{M_{r}R\}\mathbf{I}_{3}-R^{\top}M_{r})\Omega$
	\cite{hashim2019SO3Wiley}, define
	\[
	\mathcal{L}_{R}=2||M_{r}\tilde{R}||_{{\rm I}}+\frac{1}{2\gamma_{\Omega}}||\tilde{b}_{\Omega}||^{2}+\frac{\mathbf{vex}(\boldsymbol{\mathcal{P}}_{a}(M_{r}\tilde{R}))^{\top}\hat{R}\tilde{b}_{\Omega}}{2\overline{\gamma}_{\Omega}\overline{\lambda}_{\overline{M_{r}}}}
	\]
	Using \eqref{eq:UWB_lemm2_1} one finds
	\[
	e_{R}^{\top}\underbrace{\left[\begin{array}{cc}
			2 & -\frac{1}{\overline{\gamma}_{\Omega}}\\
			-\frac{1}{\overline{\gamma}_{\Omega}} & \frac{1}{2\gamma_{\Omega}}
		\end{array}\right]}_{Q_{1}}e_{R}\leq\mathcal{L}_{R}\leq e_{R}^{\top}\underbrace{\left[\begin{array}{cc}
			2 & \frac{1}{\overline{\gamma}_{\Omega}}\\
			\frac{\delta}{\overline{\gamma}_{\Omega}} & \frac{1}{2\gamma_{\Omega}}
		\end{array}\right]}_{Q_{2}}e_{R}
	\]
	where $e_{R}=[\sqrt{||M_{r}\tilde{R}||_{{\rm I}}},||\hat{R}\tilde{b}_{\Omega}||]^{\top}$.
	One shows $\frac{1}{2\overline{\gamma}_{\Omega}\overline{\lambda}_{\overline{M_{r}}}}\frac{d}{dt}(\mathbf{vex}(\boldsymbol{\mathcal{P}}_{a}(M_{r}\tilde{R}))^{\top}\hat{R}\tilde{b}_{\Omega})\leq-\frac{\sqrt{3}}{4\overline{\gamma}_{\Omega}}||\tilde{b}_{\Omega}||^{2}+c_{1}||\boldsymbol{\Upsilon}(M_{r}\tilde{R})||^{2}+c_{2}||\boldsymbol{\Upsilon}(M_{r}\tilde{R})||\,||\tilde{b}_{\Omega}||$
	where $c_{1}=\frac{\gamma_{\Omega}}{2\overline{\gamma}_{\Omega}\overline{\lambda}_{\overline{M_{r}}}}|$
	and $c_{2}=\frac{\sqrt{3}k_{\Omega}+c_{\Omega}}{2\overline{\gamma}_{\Omega}\overline{\lambda}_{\overline{M_{r}}}}$.
	Consequently,
	\begin{align}
		\dot{\mathcal{L}}_{R}\leq & -(k_{\Omega}-c_{1})||\boldsymbol{\Upsilon}(M_{r}\tilde{R})||^{2}-\frac{\sqrt{3}}{4\overline{\gamma}_{\Omega}}||\tilde{b}_{\Omega}||^{2}\nonumber \\
		& +c_{2}||\boldsymbol{\Upsilon}(M_{r}\tilde{R})||\,||\tilde{b}_{\Omega}||\label{eq:UWB_LyapR_1dot_2}
	\end{align}
	Recalling \eqref{eq:UWB_lemm2_2}, one finds
	\begin{align}
		\dot{\mathcal{L}}_{R}\leq & -e_{R}^{\top}\underbrace{\left[\begin{array}{cc}
				\frac{k_{\Omega}-c_{1}}{c_{R}} & -c_{2}\overline{\lambda}_{\overline{M_{r}}}\\
				-c_{2}\overline{\lambda}_{\overline{M_{r}}} & \frac{\sqrt{3}}{4\overline{\gamma}_{\Omega}}
			\end{array}\right]}_{Q_{2}}e_{R}\label{eq:UWB_LyapR_1dot_3}\\
		\mathcal{L}_{R}\leq & \mathcal{L}_{R}(0)\exp(-\underline{\lambda}_{Q_{3}}t/\overline{\lambda}_{Q_{2}})\label{eq:UWB_LyapRfinal}
	\end{align}
	where $c_{R}=\frac{\underline{\lambda}_{\overline{M_{r}}}}{2}(1+{\rm Tr}\{\tilde{R}(0)\})$
	and $Q_{2}$ is made positive by selecting $k_{\Omega}>\frac{4c_{R}\overline{\gamma}_{\Omega}c_{2}^{2}\overline{\lambda}_{\overline{M_{r}}}^{2}}{\sqrt{3}}+c_{1}$.
	Define the following real value function:
	\begin{align}
		\mathcal{L}_{PV}= & \frac{1}{2}||\tilde{P}||^{2}+\frac{1}{2k_{a}}||\tilde{V}||^{2}-\delta\tilde{P}^{\top}\tilde{V}\label{eq:UWB_LyapPV}
	\end{align}
	\[
	e_{PV}^{\top}\underbrace{\left[\begin{array}{cc}
			\frac{1}{2} & -\frac{\delta}{2}\\
			-\frac{\delta}{2} & \frac{1}{2k_{a}}
		\end{array}\right]}_{Q_{4}}e_{PV}\leq\mathcal{L}_{PV}\leq e_{PV}^{\top}\underbrace{\left[\begin{array}{cc}
			\frac{1}{2} & \frac{\delta}{2}\\
			\frac{\delta}{2} & \frac{1}{2k_{a}}
		\end{array}\right]}_{Q_{5}}e_{PV}
	\]
	where $e_{PV}=[||\tilde{P}||,||\tilde{V}||]^{\top}$. Hence, using
	\eqref{eq:UWB_Filter1_Error_dot} and \eqref{eq:UWB_LyapPV}, one
	obtains
	\begin{align}
		\dot{\mathcal{L}}_{PV}\leq & -(k_{v}-\delta k_{a})||\tilde{P}||^{2}-\delta||\tilde{V}||^{2}+\delta k_{v}||\tilde{V}||\,||\tilde{P}||\nonumber \\
		& +(\delta||\tilde{P}||+\frac{1}{k_{a}}||\tilde{V}||)||\mathbf{I}_{3}-\tilde{R}||_{F}\hat{R}a\label{eq:UWB_LyapPV_1dot_1}
	\end{align}
	where $c_{a}=\max\{\sup_{t\geq0}4\delta\overline{\lambda}_{M}a,\sup_{t\geq0}\frac{4\overline{\lambda}_{M}a}{k_{a}}\}$,
	$||{\rm Tr}\{M_{r}\tilde{R}\}\mathbf{I}_{3}-M_{r}\tilde{R}||_{F}\leq\sqrt{3}\overline{\lambda}_{\overline{M_{r}}}$.
	Thus, $\dot{\mathcal{L}}_{PV}$ in \eqref{eq:UWB_LyapPV_1dot_1} becomes
	\begin{align}
		\dot{\mathcal{L}}_{PV}\leq & -e_{PV}^{\top}\underbrace{\left[\begin{array}{cc}
				k_{v}-\delta k_{a} & -\frac{\delta k_{v}}{2}\\
				-\frac{\delta k_{v}}{2} & \delta
			\end{array}\right]}_{Q_{6}}e_{PV}\nonumber \\
		& +c_{a}(||\tilde{P}||+||\tilde{V}||)\sqrt{||M\tilde{R}||_{{\rm I}}}\label{eq:UWB_LyapPV_1dot_Final}
	\end{align}
	where $e_{PV}=[||\tilde{P}||^{2},||\tilde{V}||^{2}]^{\top}$. $Q_{6}$
	is made positive by selecting $\frac{4k_{v}}{k_{v}^{2}+4k_{a}}>\delta$.
	Let us define $\underline{\lambda}_{PV}=\underline{\lambda}(Q_{6})$,
	and $e_{T}=[||e_{R}||,||e_{PV}||]^{\top}$. From \eqref{eq:UWB_LyapRfinal}
	and \eqref{eq:UWB_LyapPV_1dot_Final}, one finds
	\begin{align}
		\dot{\mathcal{L}}_{T}\leq & -\underline{\lambda}_{Q_{3}}||e_{R}||^{2}-\underline{\lambda}_{Q_{6}}||e_{PV}||^{2}\nonumber \\
		& +c_{a}(||\tilde{P}||+||\tilde{V}||)\sqrt{||M\tilde{R}||_{{\rm I}}}\nonumber \\
		\leq & -e_{T}^{\top}\underbrace{\left[\begin{array}{cc}
				\underline{\lambda}_{Q_{3}} & -\frac{c_{a}}{2}\\
				-\frac{c_{a}}{2} & \underline{\lambda}_{Q_{6}}
			\end{array}\right]}_{Q_{T}}e_{T}\label{eq:UWB_LyapT_Final}
	\end{align}
	where $\eta_{\sigma}=(\frac{1}{4k_{d}}+\frac{k_{\sigma}}{2})||\sigma||^{2}$
	and $e_{T}=[||e_{R}||,||e_{PV}||]^{\top}$. Therefore, $Q_{T}$ is
	made positive by selecting $\underline{\lambda}_{Q_{3}}>\frac{c_{a}^{2}}{4\underline{\lambda}_{Q_{6}}}$.
	Thereby, $e_{T}$ is uniformly almost globally exponentially stable
	completing the proof.\end{proof}

\subsection{Accelerometer Compensation}

The proof of Theorem \ref{thm:Theorem1} can be extended to include
accelerometer compensation. Let us define $\dot{\hat{b}}_{a}$ and
modify $\dot{\hat{V}}$ as follows:
\begin{equation}
	\begin{cases}
		\dot{\hat{b}}_{a} & =-\gamma_{a}\hat{R}^{\top}(P_{y}-\hat{P})\\
		\dot{\hat{V}} & =\hat{R}(a_{m}-\hat{b}_{a})+\overrightarrow{\mathtt{g}}-\left[w_{\Omega}\right]_{\times}\hat{V}-w_{a}
	\end{cases}\label{eq:UWB_Update}
\end{equation}
where $\gamma_{a}>0$ is a positive gain. In this regard, $\mathcal{L}_{PV}$
is modified as follows: 
\[
\mathcal{L}_{PV}=\frac{1}{2}||\tilde{P}||^{2}+\frac{1}{2k_{a}}||\tilde{V}||^{2}-\delta\tilde{P}^{\top}\tilde{V}+\delta_{a}\tilde{b}_{a}^{\top}\hat{R}^{\top}\tilde{V}
\]
Analogously to the proof of Theorem \ref{thm:Theorem1}, one obtains
a result similar to \eqref{eq:UWB_LyapT_Final}.

\subsection{Implementation Steps in Discrete Form}

Let $\Delta t$ be a small sample time, and set $\hat{P}_{0|0},\hat{V}_{0|0},\hat{\sigma}_{0}\in\mathbb{R}^{3}$,
$\hat{R}_{0|0}\in\mathbb{SO}\left(3\right)$, and $k=0$. Algorithm
\ref{alg:Alg_Disc0} details the discrete implementation steps.
\begin{algorithm}[h]
	\caption{\label{alg:Alg_Disc0}Discrete navigation observer}
	
	\textbf{while }(1)\textbf{ do}
	\begin{enumerate}
		\item[{\footnotesize{}1:}] $\hat{X}_{k|k}=\left[\begin{array}{ccc}
			\hat{R}_{k|k} & \hat{P}_{k|k} & \hat{V}_{k|k}\\
			0_{1\times3} & 1 & 0\\
			0_{1\times3} & 0 & 1
		\end{array}\right]\in\mathbb{SE}_{2}\left(3\right)$ and \\
		$\hat{U}_{k}=\left[\begin{array}{ccc}
			[\Omega_{m}[k]\text{\ensuremath{]_{\times}}} & 0_{3\times1} & a_{m}[k]\\
			0_{1\times3} & 0 & 0\\
			0_{1\times3} & 1 & 0
		\end{array}\right]\in\mathcal{U}_{\mathcal{M}}$
		\item[{\footnotesize{}2:}] $\hat{X}_{k+1|k}=\hat{X}_{k|k}\exp(\hat{U}_{k}\Delta t)$
		\item[{\footnotesize{}3:}] $\begin{cases}
			v_{1}=\frac{a_{m}}{||a_{m}||}, & r_{1}=\frac{-\overrightarrow{\mathtt{g}}}{||-\overrightarrow{\mathtt{g}}||}\\
			v_{2}=\frac{m_{m}}{||m_{m}||}, & r_{2}=\frac{m_{r}}{||m_{r}||}\\
			v_{3}=\frac{v_{1}\times v_{2}}{||v_{1}\times v_{2}||}, & r_{3}=\frac{r_{1}\times r_{2}}{||r_{1}\times r_{2}||}
		\end{cases}$
		\item[{\footnotesize{}4:}] $\begin{cases}
			\overline{P}_{y} & =\left[\begin{array}{c}
				P_{y}\\
				||P-h_{1}||
			\end{array}\right]=(A^{\top}A)^{-1}A^{\top}B\\
			\hat{b}_{\Omega|k} & =\hat{b}_{\Omega|k-1}-\frac{\Delta t\gamma_{\Omega}}{2}\sum_{i=1}^{n}(v_{i}\times\hat{v}_{i})\\
			\hat{b}_{a|k} & =\hat{b}_{a|k-1}-\Delta t\gamma_{a}\hat{R}^{\top}(P_{y}-\hat{P})\\
			w_{\Omega} & =-\frac{k_{\Omega}}{2}\sum_{i=1}^{n}\hat{R}(v_{i}\times\hat{v}_{i})\\
			w_{V} & =-k_{v}(P_{y}-\hat{P})-[w_{\Omega}]_{\times}\hat{P}\\
			w_{a} & =-\overrightarrow{\mathtt{g}}-k_{a}(P_{y}-\hat{P})-[w_{\Omega}]_{\times}\hat{V}
		\end{cases}$
		\item[{\footnotesize{}5:}] $W_{k}=\left[\begin{array}{ccc}
			[w_{\Omega}[k]\text{\ensuremath{]_{\times}}} & w_{V}[k] & w_{a}[k]\\
			0_{1\times3} & 0 & 0\\
			0_{1\times3} & 1 & 0
		\end{array}\right]$
		\item[{\footnotesize{}6:}] $\hat{X}_{k+1|k+1}=\exp(-W_{k}\Delta t)\hat{X}_{k+1|k}$ and $k=k+1$
	\end{enumerate}
	\textbf{end while}
\end{algorithm}

\section{Numerical Results \label{sec:UWB_Simulations}}

In this section, effectiveness of the proposed navigation nonlinear
observer for inertial navigation using UWB and IMU is presented. The
validation utilizes a publicly available real world dataset collected
during a drone flight and published by Zhao et al., 2022 \cite{zhao2022uwbData}.
The drone was equipped with one UWB tag and a 6-axis IMU and flew
within range of 8 fixed anchors satisfying Assumption \ref{Assumption:assum_TOA}.
The dataset contains measurements of $d_{j,i}$ range, gyroscope,
and magnetometer, and fixed UWB anchor positions. The dataset also
includes the ground truth: true drone orientation (described in unit
quaternion) and position in meters. Since linear velocity is not provided,
a classical maximum likelihood (ML) approach is utilized to extract
the true linear velocity (to create a benchmark for the estimates)
\cite{maybeck1982stochastic}. Capturing the large initial error,
the experiment commenced at the true drone position $P(0)=[1.237,0.124,1.534]^{\top}$
and linear velocity $V(0)=[-0.0473,0.1286,-1.2789]^{\top}$, whereas
the estimated position and linear velocity were set to $\hat{P}(0)=[-3,-1,0]^{\top}$
and $\hat{V}(0)=[0,0,0]^{\top}$, respectively. To accommodate for
the fact that UWB tag was not placed at the drone's center, the range
distance was modified using the translation $v_{c}=[-0.012,0.001,0.091]^{\top}m$
\cite{zhao2022uwbData} as follows:
\begin{equation}
	d_{i,j}=||Rv_{c}+P-h_{j}||-||Rv_{c}+P-h_{i}||\label{eq:UWB_dji_simu}
\end{equation}
A magnetometer has been added in simulation where we defined $m_{r}=[-1.7,0,1.2]^{\top}$
and calculated $m_{m}=R^{\top}m_{r}+n_{m}$ with $n_{m}=\mathcal{N}\left(0,0.2\right)$
being a normally distributed random noise vector (zero mean and $0.2$
standard deviation). The design parameters were selected as follows:
$k_{w}=3$, $k_{v}=2$,\textbf{ }$k_{a}=70$, $\gamma_{\Omega}=0.1$,
and $\gamma_{a}=2$. Also, the initial bias estimates were set as
$\hat{b}_{\Omega}(0)=\hat{b}_{a}(0)=[0,0,0]^{\top}$.

This Section uses Trial Const1 in \cite{zhao2022uwbData}. Fig. \ref{fig:UWB_Py}
presents the true vehicle's position $P$ plotted as a red solid line,
the estimated vehicle's position $\hat{P}$ marked as a blue dash
line, and the reconstructed position $P_{y}$ (obtained from the TDOA
range measurements $d_{i,j}$) shown in orange color. Fig. \ref{fig:UWB_Py}
makes apparent the high level of uncertainties present in the reconstructed
position $P_{y}$ and the robust capability of the proposed observer
to reject the noise and provide good estimates. In Fig. \ref{fig:UWB_Err},
strong tracking performance of errors in orientation $||\tilde{R}||_{{\rm I}}=\frac{1}{4}{\rm Tr}\{\mathbf{I}_{3}-\hat{R}R^{\top}\}$,
position $||P-\hat{P}||$, and linear velocity $||V-\hat{V}||$ is
demonstrated.

\begin{figure}[h]
	\centering{}\centering\includegraphics[scale=0.33]{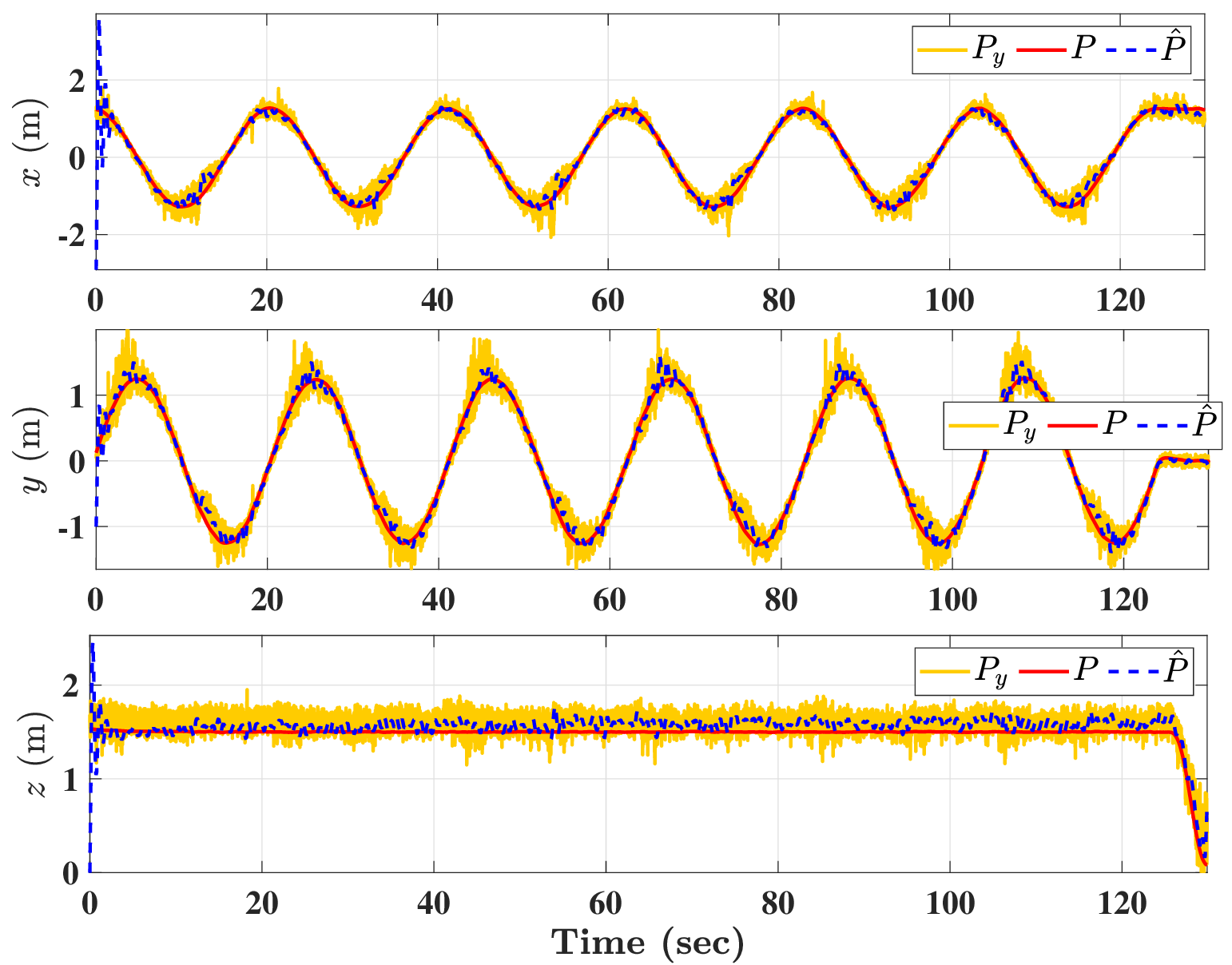}\caption{\label{fig:UWB_Py} Evolution trajectory of the true vehicle's position
		plotted as a red solid line, the estimated position depicted as a
		blue dash line, and reconstructed position marked as an orange solid
		line.}
\end{figure}

\begin{figure}[h]
	\centering{}\centering\includegraphics[scale=0.33]{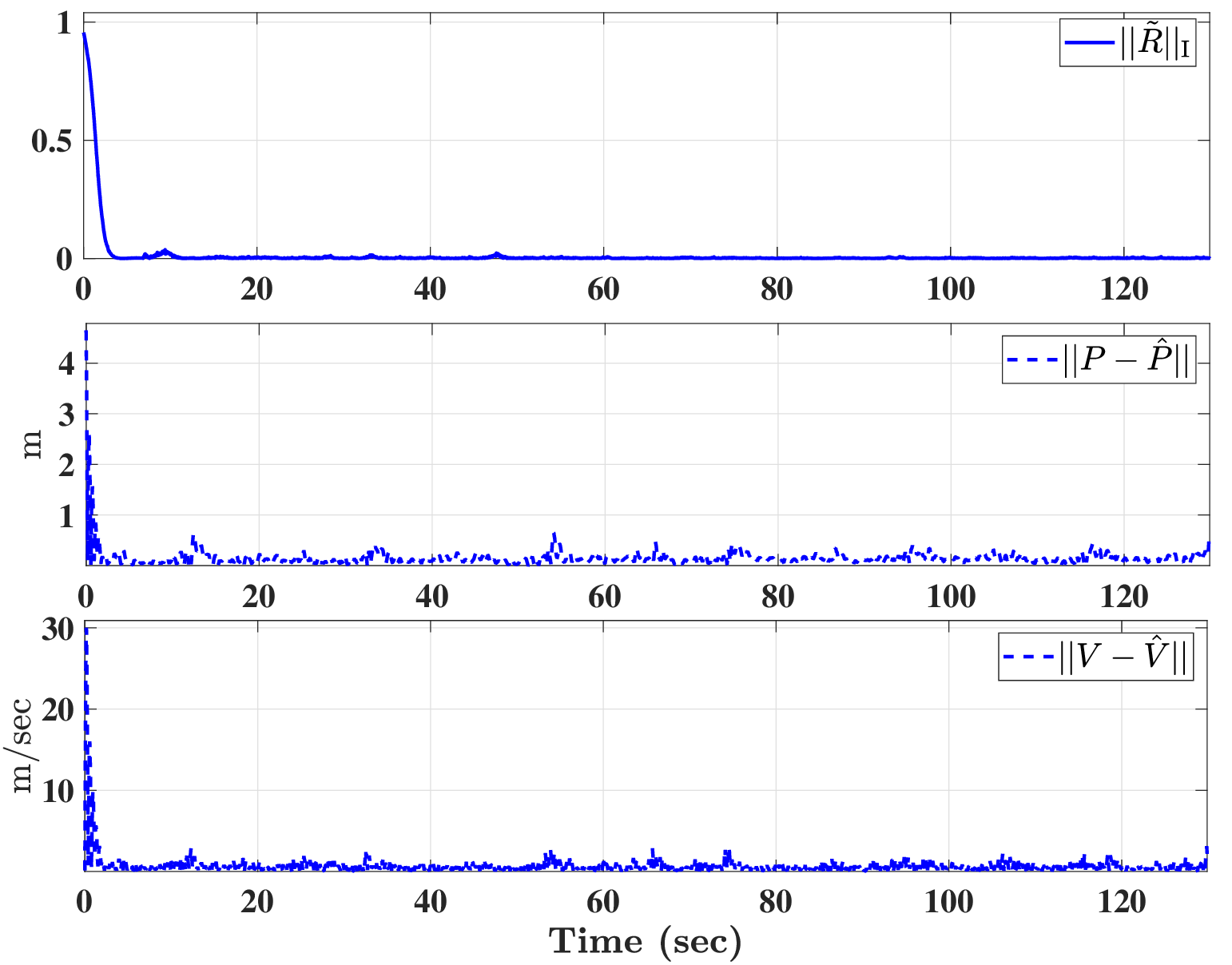}\caption{\label{fig:UWB_Err} Error convergence of orientation, position, and
		linear velocity.}
\end{figure}

\section{Conclusion \label{sec:SE3_Conclusion}}

The inertial navigation problem has been addressed using UWB-IMU fusion
to supply measurements to a nonlinear deterministic observer on the
Lie Group of $\mathbb{SE}_{2}\left(3\right)$. The observer successfully
estimates the vehicle's orientation, position, and linear velocity
ensuring exponential convergence from almost any initial condition.
The observer tackles IMU uncertainties and compensates for unknown
bias. The proposed observer revealed robust and strong estimation
performance when tested using a dataset of measurements collected
during a dataset of drone flight and benchmarked against the ground
truth. 

\section*{Acknowledgment}

The authors would like to thank \textbf{Maria Shaposhnikova} for proofreading
the article.

\noindent 
\bibliographystyle{IEEEtran}
\bibliography{bib_UWB}

\begin{thebibliography}{10}
\providecommand{\url}[1]{#1}
\csname url@samestyle\endcsname
\providecommand{\newblock}{\relax}
\providecommand{\bibinfo}[2]{#2}
\providecommand{\BIBentrySTDinterwordspacing}{\spaceskip=0pt\relax}
\providecommand{\BIBentryALTinterwordstretchfactor}{4}
\providecommand{\BIBentryALTinterwordspacing}{\spaceskip=\fontdimen2\font plus
\BIBentryALTinterwordstretchfactor\fontdimen3\font minus
  \fontdimen4\font\relax}
\providecommand{\BIBforeignlanguage}[2]{{%
\expandafter\ifx\csname l@#1\endcsname\relax
\typeout{** WARNING: IEEEtran.bst: No hyphenation pattern has been}%
\typeout{** loaded for the language `#1'. Using the pattern for}%
\typeout{** the default language instead.}%
\else
\language=\csname l@#1\endcsname
\fi
#2}}
\providecommand{\BIBdecl}{\relax}
\BIBdecl

\bibitem{li2021openstreetmap}
J.~e.~a. Li, ``Openstreetmap-based autonomous navigation for the four
  wheel-legged robot via 3d-lidar and ccd camera,'' \emph{IEEE Transactions on
  Industrial Electronics}, vol.~69, no.~3, pp. 2708--2717, 2021.

\bibitem{hashim2021_COMP_ENG_PRAC}
H.~A. Hashim, M.~Abouheaf, and M.~A. Abido, ``Geometric stochastic filter with
  guaranteed performance for autonomous navigation based on {IMU} and feature
  sensor fusion,'' \emph{Control Engineering Practice}, vol. 116, p. 104926,
  2021.

\bibitem{zhai2020robust}
C.~Zhai, M.~Wang, Y.~Yang, and K.~Shen, ``Robust vision-aided inertial
  navigation system for protection against ego-motion uncertainty of unmanned
  ground vehicle,'' \emph{IEEE Transactions on Industrial Electronics},
  vol.~68, no.~12, pp. 12\,462--12\,471, 2020.

\bibitem{zou2019structvio}
D.~Zou and et~al, ``Structvio: visual-inertial odometry with structural
  regularity of man-made environments,'' \emph{IEEE Transactions on Robotics},
  vol.~35, no.~4, pp. 999--1013, 2019.

\bibitem{hashim2021ACC}
H.~A. {Hashim}, ``Gps-denied navigation: Attitude, position, linear velocity,
  and gravity estimation with nonlinear stochastic observer,'' in \emph{2021
  American Control Conference (ACC)}.\hskip 1em plus 0.5em minus 0.4em\relax
  IEEE, 2021, pp. 1146--1151.

\bibitem{fornasier2022equivariant}
A.~Fornasier, Y.~Ng, R.~Mahony, and S.~Weiss, ``Equivariant filter design for
  inertial navigation systems with input measurement biases,'' \emph{2022 IEEE
  International Conference on Robotics and Automation (ICRA)}, 2022.

\bibitem{yang2021novel}
X.~Yang, J.~Wang, D.~Song, B.~Feng, and H.~Ye, ``A novel nlos error
  compensation method based imu for uwb indoor positioning system,'' \emph{IEEE
  Sensors Journal}, vol.~21, no.~9, pp. 11\,203--11\,212, 2021.

\bibitem{zihajehzadeh2015uwb}
S.~Zihajehzadeh and et~al, ``Uwb-aided inertial motion capture for lower body
  3-d dynamic activity and trajectory tracking,'' \emph{IEEE Transactions on
  Instrumentation and Measurement}, vol.~64, no.~12, pp. 3577--3587, 2015.

\bibitem{you2020data}
W.~You, F.~Li, L.~Liao, and M.~Huang, ``Data fusion of uwb and imu based on
  unscented kalman filter for indoor localization of quadrotor uav,''
  \emph{IEEE Access}, vol.~8, pp. 64\,971--64\,981, 2020.

\bibitem{wang2020multiple}
W.~Wang, D.~Marelli, and M.~Fu, ``Multiple-vehicle localization using maximum
  likelihood kalman filtering and ultra-wideband signals,'' \emph{IEEE Sensors
  Journal}, vol.~21, no.~4, pp. 4949--4956, 2021.

\bibitem{bottigliero2021low}
S.~Bottigliero and et~al, ``A low-cost indoor real-time locating system based
  on tdoa estimation of uwb pulse sequences,'' \emph{IEEE Transactions on
  Instrumentation and Measurement}, vol.~70, pp. 1--11, 2021.

\bibitem{tian2019resetting}
Q.~Tian, I.~Kevin, K.~Wang, and Z.~Salcic, ``A resetting approach for ins and
  uwb sensor fusion using particle filter for pedestrian tracking,'' \emph{IEEE
  Transactions on Instrumentation and Measurement}, vol.~69, no.~8, pp.
  5914--5921, 2020.

\bibitem{kallianpur2013stochastic}
G.~Kallianpur, \emph{Stochastic filtering theory}.\hskip 1em plus 0.5em minus
  0.4em\relax Springer Science, 2013.

\bibitem{lefferts1982kalman}
E.~J. Lefferts, F.~L. Markley, and M.~D. Shuster, ``Kalman filtering for
  spacecraft attitude estimation,'' \emph{Journal of Guidance, Control, and
  Dynamics}, vol.~5, no.~5, pp. 417--429, 1982.

\bibitem{zhao2022uwbData}
W.~Zhao, A.~Goudar, X.~Qiao, and A.~P. Schoellig, ``Util: An ultra-wideband
  time-difference-of-arrival indoor localization dataset,'' in
  \emph{International Journal of Robotics Research (IJRR)}, 2022.

\bibitem{hashim2018SO3Stochastic}
H.~A. Hashim, L.~J. Brown, and K.~McIsaac, ``Nonlinear stochastic attitude
  filters on the special orthogonal group 3: Ito and stratonovich,'' \emph{IEEE
  Transactions on Systems, Man, and Cybernetics: Systems}, vol.~49, no.~9, pp.
  1853--1865, 2019.

\bibitem{hashim2019SO3Wiley}
H.~A. Hashim, ``Systematic convergence of nonlinear stochastic estimators on
  the special orthogonal group {SO}(3),'' \emph{International Journal of Robust
  and Nonlinear Control}, vol.~30, no.~10, pp. 3848--3870, 2020.

\bibitem{barrau2016invariant}
A.~Barrau and S.~Bonnabel, ``The invariant extended kalman filter as a stable
  observer,'' \emph{IEEE Transactions on Automatic Control}, vol.~62, no.~4,
  pp. 1797--1812, 2016.

\bibitem{hashim2020SE3Stochastic}
H.~A. Hashim and F.~L. Lewis, ``Nonlinear stochastic estimators on the special
  euclidean group {SE}(3) using uncertain imu and vision measurements,''
  \emph{IEEE Transactions on Systems, Man, and Cybernetics: Systems}, vol.~51,
  no.~12, pp. 7587--7600, 2021.

\bibitem{kang2019unscented}
D.~Kang, C.~Jang, and F.~C. Park, ``Unscented kalman filtering for simultaneous
  estimation of attitude and gyroscope bias,'' \emph{IEEE/ASME Transactions on
  Mechatronics}, vol.~24, no.~1, pp. 350--360, 2019.

\bibitem{stovner2018attitude}
B.~N. Stovner, T.~A. Johansen, T.~I. Fossen, and I.~Schjolberg, ``Attitude
  estimation by multiplicative exogenous kalman filter,'' \emph{Automatica},
  vol.~95, pp. 347--355, 2018.

\bibitem{maybeck1982stochastic}
P.~S. Maybeck, \emph{Stochastic models, estimation, and control}.\hskip 1em
  plus 0.5em minus 0.4em\relax Academic press, 1982.

\end{thebibliography}
\end{document}